\documentclass{llncs}
\usepackage{hyperref}
\usepackage{listings}
\usepackage{xcolor}
\usepackage{amsmath,amssymb}
\usepackage[T1]{fontenc}
\usepackage{booktabs}                 
\usepackage{multirow}                 
\usepackage{graphicx}                 

\makeatletter
\let\@oldthebibliography\thebibliography
\renewcommand{\thebibliography}[1]{%
  \@oldthebibliography{#1}%
  \setlength{\itemsep}{0pt}\setlength{\parsep}{0pt}\setlength{\parskip}{0pt}}
\makeatother
\usepackage{algorithm}                
\usepackage{algpseudocode}            
\usepackage{cleveref}

\newif\ifextended
\extendedtrue

\newcommand{\artefacturl}{\url{https://github.com/PRETgroup/rv-artifact-threevalued-observers}}

\usepackage{fontawesome}

\usepackage{comment}
\newif\ifshowsuggestions
\showsuggestionsfalse
\definecolor{suggred}{RGB}{225,0,0}
\ifshowsuggestions
  \newcommand{\suggest}[1]{{\color{suggred}#1}}
  \specialcomment{suggestion}{\begingroup\color{suggred}}{\endgroup}
\else
  \newcommand{\suggest}[1]{}
  \excludecomment{suggestion}
\fi

\makeatletter
\@ifundefined{definition}{\spnewtheorem{definition}{Definition}{\bfseries}{\itshape}}{}
\@ifundefined{lemma}{\spnewtheorem{lemma}{Lemma}{\bfseries}{\itshape}}{}
\@ifundefined{corollary}{\spnewtheorem{corollary}{Corollary}{\bfseries}{\itshape}}{}
\@ifundefined{proposition}{\spnewtheorem{proposition}{Proposition}{\bfseries}{\itshape}}{}
\@ifundefined{property}{\spnewtheorem{property}{Property}{\bfseries}{\itshape}}{}
\@ifundefined{example}{\spnewtheorem{example}{Example}{\bfseries}{\rmfamily}}{}
\@ifundefined{remark}{\spnewtheorem{remark}{Remark}{\itshape}{\rmfamily}}{}
\makeatother
\providecommand{\llbracket}{[\![}
\providecommand{\rrbracket}{]\!]}

\lstdefinestyle{lustre}{
  basicstyle=\ttfamily\fontsize{8.5pt}{9.5pt}\selectfont,
  keywordstyle=\bfseries,
  commentstyle=\color{green!60!black},
  morecomment=[l]{--},
  morekeywords={node,returns,var,let,tel,true,false,pre,and,or,not,if,then,else,int,bool},
  frame=single,
  breaklines=true,
  postbreak=\mbox{\textcolor{red}{$\hookrightarrow$}\space},
}

\title{Synchronous Observers Revisited for Runtime Verification of Lustre Using STL}
\titlerunning{Synchronous Observers Revisited for RV of Lustre Using STL}
\authorrunning{L. Kenwright et al.}

\author{Logan Kenwright\inst{1} \and
	Partha Roop\inst{1} \and
	Sobhan Chatterjee\inst{1} \and
	Nathan Allen\inst{2}
}

\institute{The University of Auckland, Auckland, New Zealand\\
	\email{\{logan.kenwright, p.roop, sobhan.chatterjee\}@auckland.ac.nz} \and
	Auckland University of Technology, Auckland, New Zealand\\
	\email{nathan.allen@aut.ac.nz}
}

\begin{document}

\maketitle

\begin{abstract}
Signal Temporal Logic (STL) is a popular formalism for the temporal safety properties of cyber-physical systems, most often used for runtime verification. In the synchronous family of languages, safety properties are instead expressed as synchronous observers, modules composed with a program for static verification, which are also runnable specifications suitable for runtime verification, though this use is rarely explored. We present a technique for compiling the synchronous fragment of STL (SSTL) into synchronous observers in the dataflow language Lustre. Unlike previous work, we allow arbitrary nesting of bounded SSTL properties via modular compilation, and admit a globally unbounded outer operator for online monitoring; the resulting observers serve both runtime verification and, as a by-product, static verification with the Kind2 model checker. We further contribute an interactive visualiser that renders a property's three-valued verdict over an editable trace, and evaluate on two case studies from the literature: a spring-mass system and a car-following cruise controller.

\end{abstract}

\ifextended
\noindent\textit{This is the extended version of a paper accepted at Runtime
Verification (RV 2026). It is identical to the proceedings version except for
\Cref{app:proofs}, which contains the proofs omitted there for space. The final
authenticated version will appear in the Springer LNCS proceedings.}
\fi

\section{Introduction}
\label{sec:introduction}
Many cyber-physical systems are safety-critical, so it is of great interest to specify
their safety properties formally. One popular tool for this is Signal Temporal Logic
(STL), whose main feature is the ability to specify time intervals under which a
condition is expected to hold. STL is most commonly applied for \textit{runtime
verification}, where properties are evaluated during the actual deployment of a system
and a violation may trigger a halt or other corrective action. This is convenient for
improving the safety of complex black-box systems without needing a formal model of the
system itself.

Existing STL tooling largely monitors pre-recorded or simulated signals, or operates on abstract automaton models, rather than on the executable code that is actually deployed; there is
consequently no executable, runtime-equivalent semantics of \emph{nested} STL that
runs in lock-step with the system it monitors. We close this gap by synthesising the property directly into a synchronous observer in the same language as the system, in this work using the dataflow language \textit{Lustre}. Because the observer shares the system's formal semantics, the very same artefact serves two purposes: it monitors the system online, and as a bonus can be discharged by a model checker before deployment, with no separate model to build.

We make three contributions:
\begin{enumerate}
  \item \textbf{Online monitoring of a truly executable CPS.} The observer is compiled
    together with the system in the same synchronous language and runs in lock-step with
    it, rather than over pre-recorded or simulated signals as most existing STL tooling
    does (\Cref{sec:related}).
  \item \textbf{A three-valued streaming semantics for runtime.} Causal Kleene verdicts
    (\textsc{def\_true}/\textsc{def\_false}/\textsc{unknown}) let the monitor commit to a
    definitive answer as early as the trace allows, ahead of the formula horizon.
  \item \textbf{A visualiser for nested STL formulae.} An interactive, browser-based tool
    renders a nested observer's three-valued verdict over an editable trace, making the
    otherwise unintuitive behaviour of nested operators legible.
\end{enumerate}

The synthesiser, the visualiser, and both case studies are available as an
open-source artefact.\footnote{\artefacturl}

The closest prior work is Bellanger et al.~\cite{bellanger2025formally}, who implement
non-nested bounded STL operators as synchronous observers. Relative to them, our
technique is novel in three ways:

\begin{itemize}
    \item We allow arbitrary nesting of bounded SSTL operators, enabling the specification of more complex properties
    \item Where Bellanger et al.\ orient their observers toward static
    verification as Kind2 contracts, ours are language-native executable
    specifications: the same artefact runs as an online monitor at runtime, and as a
    by-product remains amenable to static verification in Kind2.
    \item We enable the outer operator to be the \textit{Unbounded Globally} operator, for monitoring free-running systems.
\end{itemize}
The remainder of the paper is structured as follows: Sect.~\ref{sec:background} gives preliminaries on synchronous observers and synchronous signal temporal logic. Sect.~\ref{sec:casestudy} introduces two case studies with intuitive properties. Sect.~\ref{sec:synthesis} describes our observer synthesis technique, and Sect.~\ref{sec:shift} extends it to the unbounded globally operator. Sect.~\ref{sec:visualiser} presents the interactive visualiser for nested observers. Sect.~\ref{sec:evaluation} evaluates our case studies using this technique, Sect.~\ref{sec:related} surveys related work, and Sect.~\ref{sec:conclusion} concludes.

\section{Background}
\label{sec:background}

\subsection{Synchronous Signal Temporal Logic}

Signal temporal logic~\cite{donze2013signal,maler2004monitoring,bartocci_specification-based_2018} is a formal language used to describe temporal properties over predicates on a set of real-valued signals. We write $s$ for such a signal and $t\in\mathbb{R}_{\ge0}$ for a time instant, so that $(s,t)\models\phi$ denotes that $\phi$ holds for $s$ at $t$. An atomic predicate $p$ is a Boolean condition on the signal value, for example $\mathit{gap}\ge6$. The available operators are as described in \Cref{def:stl-syntax}, then their semantics in \Cref{def:stl-semantics}.

\begin{definition}[STL Syntax]\label{def:stl-syntax}
\[
  \phi ::= p \mid \neg\phi \mid \phi \wedge \phi \mid
  \square_{[a,b]}\phi \mid \diamondsuit_{[a,b]}\phi \mid
  \psi\,\mathbf{U}_{[a,b]}\,\phi
\]
where $p$ is an atomic predicate over signal values and $[a,b]$ is a bounded time interval. $\square_{[a,b]}$ (always) declares that a predicate must remain true for the inclusive time interval, and $\diamondsuit_{[a,b]}$ (eventually) dictates that a predicate must become true at least once in the inclusive time interval. Negation and conjunction have their standard mathematical definitions. $\psi\,\mathbf{U}_{[a,b]}\,\phi$ (until) declares that the precondition $\psi$ must hold continuously until the postcondition $\phi$ is satisfied.
\end{definition}

\begin{definition}[STL Boolean Semantics]\label{def:stl-semantics}
  \begin{itemize}
    \item $\square_{[a,b]}\phi$: $\forall_{t'\in[t{+}a,t{+}b]} (s,t')\models\phi$.
    \item $\diamondsuit_{[a,b]}\phi$: $\exists_{t'\in[t{+}a,t{+}b]} (s,t')\models\phi$.
    \item $\psi\,\mathbf{U}_{[a,b]}\,\phi$: $\exists_{t'\in[t{+}a,t{+}b]} \left( (s,t')\models\phi \wedge \forall_{t'' \in [t,t')} (s,t'')\models\psi \right)$.
  \end{itemize}
\end{definition}

Moreover, in this work we are concerned only with the discrete projection of a continuous signal, divided into the logical ticks of execution which synchronous programs use. This is called Synchronous Signal Temporal Logic (SSTL)~\cite{roop2026synchronoussignaltemporallogic}.

\begin{definition}[Synchronous STL~\cite{roop2026synchronoussignaltemporallogic}]\label{def:sstl}
SSTL is the discrete-time abstraction of STL in which the signal is sampled at fixed
steps called \emph{ticks}, so time ranges over $\mathbb{N}$ rather than
$\mathbb{R}_{\ge0}$. Its syntax is that of \Cref{def:stl-syntax} with every interval
bound $[a,b]$ an integer number of ticks, and its semantics is that of
\Cref{def:stl-semantics} with the real time $t$ replaced by the tick index $n$; each
atomic predicate is thus a Boolean test on the sampled signal, such as
$\mathit{safe\_gap}\equiv(\mathit{gap}\ge6)$. Under certain robustness conditions, SSTL and STL agree on satisfaction, making SSTL a sound and complete discrete abstraction of STL.
\end{definition}

To make these operators concrete, we draw on the adaptive cruise control (ACC) case
study used throughout this paper: a follower vehicle regulates its distance to a lead
vehicle that brakes to a stop. Let $\mathit{gap}$ be the inter-vehicle distance and
define the predicate $\mathit{sg}\equiv(\mathit{gap}\ge 6\,\mathrm{m})$, a safe
following gap. Single operators already express useful requirements:
$\square_{[0,b]}\,\mathit{sg}$ asks that the gap stay safe across an interval of $b$
ticks, and $\diamondsuit_{[1,4]}\,\mathit{sg}$ that, if the gap becomes unsafe, it
recovers to safe within four ticks. Composing them,
$\square_{[2,6]}\bigl(\diamondsuit_{[1,4]}\,\mathit{sg}\bigr)$ asks that at every tick
of a start-up window the gap recovers within four ticks.

The relationship between a true continuous signal $s[t]$  and the sampled trace $\hat{s}[n]$ actually available to a monitor, is already well handled in the literature through sampling theory and the concept of \textit{robustness}, and is outside the scope of this paper.

\subsection{Synchronous Languages and Lustre}

Synchronous languages divide computation
into discrete instants, called logical
ticks. Of the classical members~\cite{berry1992esterel,leguernic1991signal,halbwachs1991lustre}, the dataflow
language Lustre sees the most use today, particularly through its industrial variant
SCADE~\cite{scade} in aviation. Their strong formal semantics make them amenable to static
verification and to deterministic execution. As cyber-physical systems they rely on
the \textit{synchrony hypothesis}: if a tick executes faster than the inter-arrival
time of environmental signals, it may be treated as instantaneous, so the effect of
physical time is abstracted away. A Lustre node is a set of equations evolving over
ticks; for example, the lead vehicle of our case study cruises at
$20\,\mathrm{m/s}$ and then brakes at $-8\,\mathrm{m/s^2}$ from tick~$6$ until it
stops, integrating its motion over ticks of $dt = 0.5\,\mathrm{s}$:

\noindent\begin{minipage}{\linewidth}
\begin{lstlisting}[style=lustre]
node lead_vehicle () returns (pos1, vel1: real);
var n: int; a1: real;
let
  n    = 0 -> pre(n) + 1;
  a1   = if n >= 6 and (20.0 -> pre(vel1)) > 0.0 then -8.0 else 0.0;
  vel1 = 20.0 -> if pre(vel1) + a1 * 0.5 < 0.0 then 0.0
                 else pre(vel1) + a1 * 0.5;
  pos1 = 0.0 -> pre(pos1) + pre(vel1) * 0.5;
tel
\end{lstlisting}
\end{minipage}
Here \texttt{vel1} and \texttt{pos1} are the outputs and \texttt{n}, \texttt{a1}
internal variables. The \texttt{->} operator gives a stream's value on the first tick
and switches to its right-hand side thereafter, while \texttt{pre} reads the
previous-tick value.

\subsection{Synchronous Programs and Observers}
Most synchronous programs can be formalised as a Mealy machine, following the
definition used by Argos~\cite{maraninchi2001argos}.

\begin{definition}[Synchronous Program]\label{def:sync-program}
A synchronous program $P$ is a Mealy machine $(Q, q_0, I, O, \delta)$ with states
$Q$, initial state $q_0$, finite sets of input and output variables $I$ and
$O$, and transition function $\delta: Q \times \mathrm{Val}(I) \to Q \times
\mathrm{Val}(O)$, where $\mathrm{Val}(X)$ is the set of valuations of the variables $X$. At each tick it reads one or more inputs,
$\sigma\in\mathrm{Val}(I)$, updates its state by $\delta$, and emits outputs $o\in\mathrm{Val}(O)$.
\end{definition}

An observer is another synchronous program composed in parallel with $P$, watching
its inputs and outputs and emitting an alarm when a condition is met; the alarm is
consumed by a runtime monitor or model checker.

\begin{definition}[Synchronous Observer]\label{def:sync-observer}
  A synchronous observer of $P$ is a synchronous program $P_\phi = (Q_\phi,
  q_{\phi,0}, \Sigma, O_\phi, \delta_\phi)$ with $\Sigma = I \cup O$ and
  $O_\phi = \{alarm, \bot\}$. The executable system is the parallel composition
  $P \,||\, P_\phi$.
\end{definition}

\subsection{Three-Valued (Kleene) Streaming Semantics}
\label{sec:three-valued}

An online observer cannot always decide a property at the tick where it is anchored,
because SSTL operators are forward-looking. Consider $\square_{[0,5]}p$: a single false
value, say at $t{+}1$, settles it false at once, but while $p$ has held and the window
is still open, no verdict is justified until tick $t{+}5$. This delay affects every
bounded operator and compounds under nesting, where an inner subformula pushes the
horizon beyond the outer operator's. We therefore adopt a three-valued streaming
semantics in the style of LTL$_3$~\cite{bauer2011runtime}: at each tick the observer
emits \textsc{def\_true}, \textsc{def\_false}, or \textsc{unknown}, the last meaning the
property cannot yet be decided from the prefix observed so far.

Returning to the ACC property
$\square_{[2,6]}\bigl(\diamondsuit_{[1,4]}\,\mathit{sg}\bigr)$ of
\Cref{sec:background}, synthesising our observer and running it on a simulated braking
trace (via the Lustre compiler \texttt{lv6}) yields a verdict that stays
\textsc{unknown} until tick~$7$, where it settles \textsc{def\_true}, definitive
before the formula's horizon of ten ticks. We return to this system in
\Cref{sec:casestudy}.

\begin{definition}[Three-Valued SSTL Semantics]\label{def:three-valued}\label{def:three-valued-sem}
At each tick the observer carries a pair
$\llbracket\varphi\rrbracket=(\llbracket\varphi\rrbracket^{+},\llbracket\varphi\rrbracket^{-})\in\mathbb{B}^2$,
where $\llbracket\varphi\rrbracket^{+}$ ($\mathit{pos}$) asserts $\varphi$ is
definitively true and $\llbracket\varphi\rrbracket^{-}$ ($\mathit{neg}$) that it is
definitively false; the brackets stress these are verdicts \emph{computed from}
$\varphi$. Both are defined inductively, with base case
$\llbracket p\rrbracket^{+}_{t}=p_t$ and $\llbracket p\rrbracket^{-}_{t}=\neg p_t$ for an
atom $p$. Writing $u$ for the prefix observed so far and $u\cdot\sigma$ for an arbitrary
infinite continuation, a verdict is definitive exactly when $u$ forces the Boolean
outcome regardless of continuation:
\[
  \mathcal{V}(\varphi,t)=
  \begin{cases}
    \textsc{def\_true}  & \text{if } \mathit{pos},
      \quad\text{i.e. } \forall\sigma:\ (u\cdot\sigma,t)\models\varphi,\\
    \textsc{def\_false} & \text{if } \mathit{neg},
      \quad\text{i.e. } \forall\sigma:\ (u\cdot\sigma,t)\not\models\varphi,\\
    \textsc{unknown}    & \text{otherwise,}
  \end{cases}
\]
where $\models$ is the Boolean SSTL semantics of \Cref{def:stl-semantics}.

\textsc{unknown} thus means the formula horizon has not yet
  elapsed, and collapses to a Boolean once it passes.
\end{definition}

\begin{property}[Kleene Mutual-Exclusion Invariant]\label{prop:kleene-invariant}
For every $\varphi$, prefix $u$ and tick $t$, $\mathit{pos}\wedge\mathit{neg}=
\mathit{false}$: never both \textsc{def\_true} and \textsc{def\_false}.
\end{property}

The $(\mathit{pos},\mathit{neg})$ pair and the three-valued domain
$\{\textsc{def\_true},\textsc{def\_false},\textsc{unknown}\}$ carry the same
information, by Property~\ref{prop:kleene-invariant}; using the pair makes logical negation simpler.

Each subformula is observed through its $(\mathit{pos},\mathit{neg})$ pair, so the
Boolean connectives are lifted to act on pairs by Kleene's strong three-valued
logic~\cite{kleene2002mathematical}, in which \textsc{unknown} propagates only when it
could still change the outcome. Conjunction is \textsc{def\_false} as soon as one
conjunct is and \textsc{def\_true} only once both are; disjunction is dual,
\textsc{def\_true} as soon as one disjunct is and \textsc{def\_false} only once both
are; and negation simply \emph{swaps} the two components:

\[
  \llbracket\varphi\wedge\psi\rrbracket
    = \bigl(\,\llbracket\varphi\rrbracket^{+}\wedge\llbracket\psi\rrbracket^{+},\ \
             \llbracket\varphi\rrbracket^{-}\vee  \llbracket\psi\rrbracket^{-}\,\bigr),
\]
\[
  \llbracket\varphi\vee\psi\rrbracket
    = \bigl(\,\llbracket\varphi\rrbracket^{+}\vee  \llbracket\psi\rrbracket^{+},\ \
             \llbracket\varphi\rrbracket^{-}\wedge\llbracket\psi\rrbracket^{-}\,\bigr).
\]
\[
  \llbracket\neg\varphi\rrbracket = \bigl(\,\llbracket\varphi\rrbracket^{-},\ \llbracket\varphi\rrbracket^{+}\,\bigr).
\]
Negation thus maps \textsc{def\_true}/\textsc{def\_false} to each other and
\textsc{unknown} ($\neg\mathit{pos}\wedge\neg\mathit{neg}$) to itself, costing no extra
state and preserving mutual exclusion. The temporal operators all reduce to one
primitive, the bounded until, via the standard identities of \Cref{def:derived-ops}.

\begin{definition}[Derived Operators]\label{def:derived-ops}
\begin{align}
  \diamondsuit_{[a,b]}\varphi &\;\equiv\; \top\,\mathbf{U}_{[a,b]}\,\varphi,
  \label{eq:ev-from-u}\\
  \square_{[a,b]}\varphi      &\;\equiv\; \neg\,\diamondsuit_{[a,b]}\,\neg\varphi
  \;\equiv\; \neg\bigl(\top\,\mathbf{U}_{[a,b]}\,\neg\varphi\bigr).
  \label{eq:al-from-u}
\end{align}
\end{definition}

In this work, we consider only the bounded fragment of SSTL, so we only need to implement the bounded until operator as a synchronous observer. For practicality in online systems, we include one special case where the outer operator may be the unbounded globally operator $\square$, which will be discussed later.

The observer cannot quantify over the infinitely many continuations of
\Cref{def:three-valued}; instead it \emph{enumerates} the bounded window, substituting
a child's verdict pair wherever an atom would sit, which is what lets the operators
nest.

\begin{definition}[Explicit (Enumerated) Verdicts]\label{def:explicit-semantics}
For the bounded until $\psi\,\mathbf{U}_{[a,b]}\,\varphi$ the two verdicts are
\begin{align}
  \llbracket\psi\,\mathbf{U}_{[a,b]}\,\varphi\rrbracket^{+}_{t}
    &= \bigvee_{i=a}^{b}\Bigl(\llbracket\varphi\rrbracket^{+}_{t+i}\ \wedge\
        \bigwedge_{j=0}^{i-1}\llbracket\psi\rrbracket^{+}_{t+j}\Bigr),
    \label{eq:until-pos}\\[2pt]
  \llbracket\psi\,\mathbf{U}_{[a,b]}\,\varphi\rrbracket^{-}_{t}
    &= \underbrace{\bigvee_{k=0}^{b}\Bigl(\llbracket\psi\rrbracket^{-}_{t+k}\ \wedge\
        \bigwedge_{i=a}^{\min(k,b)}\llbracket\varphi\rrbracket^{-}_{t+i}\Bigr)}_{\text{early: precondition fails}}
    \ \vee\
    \underbrace{\bigwedge_{i=a}^{b}\llbracket\varphi\rrbracket^{-}_{t+i}}_{\text{window expiry}}.
    \label{eq:until-neg}
\end{align}
\end{definition}

The positive verdict (\ref{eq:until-pos}) is a \emph{witness}: $\varphi$ holds at some
offset in the window while $\psi$ held on every earlier tick, and closes as soon as that
witness appears. The negative verdict (\ref{eq:until-neg}) closes \textsc{def\_false} in
two ways: \emph{early termination}, where $\psi$ fails before any witness (definitive
immediately, possibly long before $t{+}b$), and \emph{window expiry}, where $\psi$ held
but $\varphi$ was never witnessed across $[a,b]$ (definitive only at the horizon
$t{+}b$). Until one fires, the verdict is \textsc{unknown}.

We illustrate the three closing conditions on a single bounded until
$\psi\,\mathbf{U}_{[2,4]}\,\varphi$ applied to a real signal $x$, binarised by the
postcondition $\varphi\equiv(x\ge 6)$ and the precondition $\psi\equiv(x\ge 0)$; the
response window is the tick range $[2,4]$ and the horizon is $t{=}4$.
\Cref{fig:binarisation} samples the continuous signal at each tick, thresholds it into
$\varphi$ and $\psi$, and reads the verdict from the leaf's $(\mathtt{rp},\mathtt{rn})$
outputs under \texttt{lv6}. In~(a) $x$ crosses the threshold inside the window while
$\psi$ holds, so the verdict closes \textsc{def\_true} at the witness tick $t{=}3$, one
tick before the horizon. In~(b) $\psi$ breaks at $t{=}2$ before any witness; the leaf
detects the broken precondition at $t{=}3$ and closes \textsc{def\_false} there, ahead
of the horizon. In~(c) $\psi$ holds throughout but $x$ never reaches the threshold, so
the absence becomes conclusive only as the window closes, turning \textsc{def\_false}
exactly at the horizon $t{=}4$.

\begin{figure}[htpb]
\centering
\includegraphics[width=0.54\linewidth]{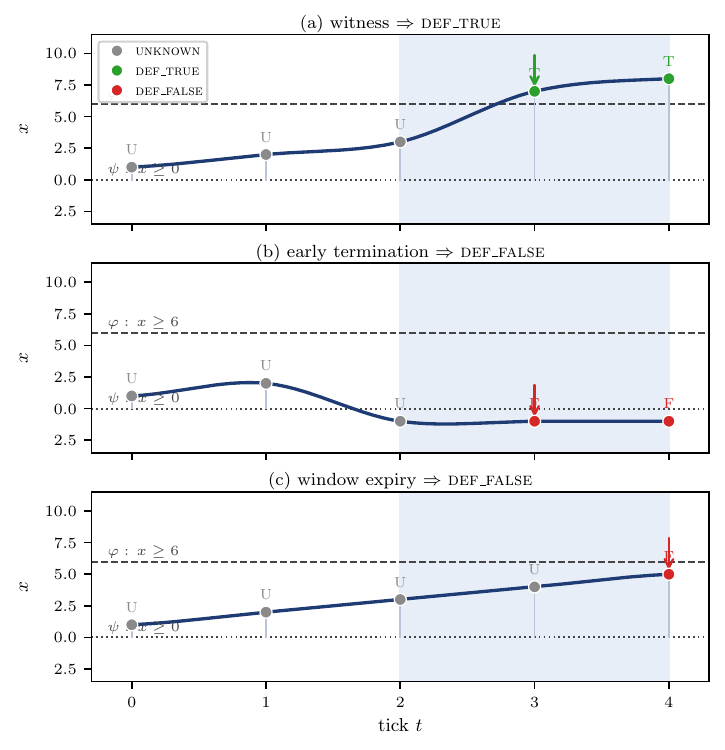}
\caption{A continuous signal $x$ sampled at discrete ticks and binarised into predicates. Dot colour is the verdict and the arrow marks the tick it becomes definitive.}
\label{fig:binarisation}
\end{figure}

\begin{remark}[Proof obligations]\label{rem:obligations}
We propose that an implementation of three-valued streaming semantics for nested SSTL properties is correct if the following obligations are met:
  \begin{enumerate}
    \item the implementation of the bounded until operator emits the explicit verdicts of \Cref{def:explicit-semantics} at each tick.
    \item substituting a child's $(\mathit{pos},\mathit{neg})$ verdict for an atom is sound, so nesting composes
  \end{enumerate}
\end{remark}

Regarding the aforementioned extension to the unbounded globally operator, we note that satisfaction is not possible in online systems, as it requires knowing an infinite future, but violation is possible. Hence, it is useful as a falsification operator for free-running systems.

\section{Motivating Case Study}
\label{sec:casestudy}

We demonstrate observer generation on two case studies. The first is a spring-mass
system following Bellanger et al.~\cite{bellanger2025formally}, who pose a nested
$\square(\diamondsuit)$ formula as a property beyond their non-nested fragment, which
our synthesis verifies directly. The second is an original adaptive cruise control
(ACC) scenario with emergency braking, automotive examples being common safety-critical
CPS benchmarks in the STL domain~\cite{roehm2016stl}. Both plant models are minimal
representative instances written for this paper, not reproductions of a published
benchmark.

\subsection{Spring-Mass System}
The system is a damped second-order plant driven toward a constant reference $r=0.3$.
\begin{align*}
  \mathit{vel}(t) &= 0.7\cdot\mathit{vel}(t{-}1) + 0.09\cdot(r - \mathit{pos}(t{-}1)),\\
  \mathit{pos}(t) &= \mathit{pos}(t{-}1) + \mathit{vel}(t{-}1),
\end{align*}
with $\mathit{pos}(0) = \mathit{vel}(0) = 0$.
For this particular configuration, the spring has a
stable, underdamped response: the position reaches the reference near tick~7, overshoots
to $\approx 0.39$ at tick~11, and converges with decreasing oscillations.

We use four Boolean predicates over the position $x$ and reference $r$, namely three
closeness tolerances and a monotonicity test ($\mathit{mono}$ taken true at $t{=}0$):
\begin{align*}
  \mathit{reach}(x,r)   &\;\triangleq\; |x - r| \le 0.05, &
  \mathit{weps}(x,r)    &\;\triangleq\; |x - r| \le 0.10,\\
  \mathit{wstrict}(x,r) &\;\triangleq\; |x - r| \le 0.02, &
  \mathit{mono}(x)      &\;\triangleq\; x \ge \mathit{pre}(x).
\end{align*}

The predicates are Boolean variables updated each tick. We monitor four nested
properties. \textbf{P1},
$\square_{[2,4]}(\diamondsuit_{[3,5]}\,\mathit{reach})$, is the 2-deep bounded
$\square(\diamondsuit)$ nesting Bellanger et al.~\cite{bellanger2025formally} place
beyond scope (their theoretical example
$\square_{[2,4]}((\diamondsuit_{[3,5]}\varphi)\wedge\psi)$), which our synthesis handles
directly. \textbf{P2},
$\square(\diamondsuit_{[5,20]}(\square_{[1,5]}\,\mathit{weps}))$, asks that from every
tick the system settles within $\pm10\%$ for five consecutive ticks inside twenty;
its unbounded outer $\square$ admits only a violation at runtime. \textbf{P2$'$}
tightens this to $\mathit{wstrict}$ ($\pm2\%$) over a shorter window, which the
underdamped transient cannot sustain, so a violation is detected, showing the
satisfaction results are not specious. \textbf{P3},
$\mathit{mono}\,\mathbf{U}_{[3,15]}(\diamondsuit_{[1,5]}(\square_{[1,3]}\,\mathit{weps}))$,
requires the position to rise monotonically until the inner settling formula is
witnessed, exercising $\mathbf{U}$ as the outermost operator.

\subsection{Car-Following Under Emergency Braking}

We study a second scenario: two vehicles on a straight road
with a discrete-time step $\Delta t = 0.5\,\mathrm{s}$.
The \emph{lead vehicle} drives at $20\,\mathrm{m/s}$ for ticks $0$--$5$, then applies
maximum emergency braking ($-8\,\mathrm{m/s^2}$) until stopped.
The \emph{follower} runs an Adaptive Cruise Control (ACC) targeting a one-second gap (minimum $6\,\mathrm{m}$):
\begin{align*}
  g(t) &= p_1(t) - p_2(t) - 4,\\
  a_2(t) &= \operatorname{clip}\!\left(
    (g(t) - 6 - v_2(t{-}1)) - 0.5\,(v_2(t{-}1) - v_1(t{-}1)),\;
    {-8},\; 2\right),
\end{align*}
where $p_i$, $v_i$ are position and velocity and vehicle length is $4\,\mathrm{m}$.
The vehicles are initialised at a gap of $g(0) = 26\,\mathrm{m}$, both at
$20\,\mathrm{m/s}$.  The follower decelerates in response to the emergency stop,
reaching a stationary gap of $6\,\mathrm{m}$.

Four boolean predicates are evaluated each tick:
\begin{align*}
  \mathit{safe\_gap}(g)    &\;\triangleq\; g \ge 6\,\mathrm{m}, &
  \mathit{adeq\_gap}(g)    &\;\triangleq\; g \ge 10\,\mathrm{m},\\
  \mathit{gap\_settled}(g) &\;\triangleq\; 6 \le g \le 12\,\mathrm{m}, &
  \mathit{vel\_moving}(v_2) &\;\triangleq\; v_2 > 2\,\mathrm{m/s}.
\end{align*}

The four properties mirror the spring-mass set. \textbf{Q1},
$\square_{[2,6]}(\diamondsuit_{[1,4]}\,\mathit{safe\_gap})$, asks the safe gap to recur
within four ticks of every start, which the ACC achieves even under braking.
\textbf{Q2}, $\square(\diamondsuit_{[3,10]}(\square_{[1,4]}\,\mathit{safe\_gap}))$, asks
it to hold for four consecutive ticks from every start; since the ACC keeps $g\ge6\,$m
throughout, the unbounded outer $\square$ holds. \textbf{Q2$'$} tightens this to the
$10\,$m adequate gap, which the emergency stop drives the gap below from tick~11, so the
violation cannot recover and is detected. \textbf{Q3},
$\mathit{vel\_moving}\,\mathbf{U}_{[2,10]}(\diamondsuit_{[1,3]}(\square_{[1,2]}\,\mathit{gap\_settled}))$,
requires the follower to keep moving until the gap settles into $[6,12]\,$m,
again with $\mathbf{U}$ outermost.

\section{N-deep Observer Synthesis}
\label{sec:synthesis}
In many interval logics, nesting is considered complex and excluded even in the
bounded setting. In our synchronous setting, however, intervals are sequences of
discrete integer ticks rather than dense continuous time, so a nested formula unrolls
into a finite collection of non-nested obligations, each watching a fixed window of
ticks. This is the key insight behind our technique: we unroll a nested formula into
these obligations and wire them together. The naive unrolling is exponential in
nesting depth, but most leaves overlap and could be shared by a more careful
construction, which we leave to future work as the present technique already suffices
for our case studies. We synthesise a single reusable Lustre node for the elementary
obligation---the bounded until, as the \texttt{until\_leaf} primitive---and build an
observer for any nested formula by instantiating and wiring copies.

\subsection{Unrolling a nested formula}

We first develop the construction for a single nested chain of bounded operators
applied to one inner predicate,
\[
  \mathit{op}_1^{[a_1,b_1]}\bigl(\mathit{op}_2^{[a_2,b_2]}(\cdots
  (\mathit{op}_N^{[a_N,b_N]}\,\phi)\cdots)\bigr),
  \qquad \mathit{op}_k \in \{\square,\diamondsuit,\mathbf{U}\}.
\]
We later extend this to several predicates combined with the Boolean connectives
$\wedge$ and $\vee$, which the same aggregation scheme handles without changing the
primitive.

Consider the running example $\square_{[2,4]}(\diamondsuit_{[3,5]}\,\phi)$. The
outer operator $\square_{[2,4]}$ does not fix a single tick. Instead, it requires
the inner formula to hold once for each start offset $i_1 \in [2,4]$. 
The formula reduces to one obligation for each starting point, and we call each such obligation a
\emph{leaf}. A leaf fixes a single absolute window. 

Two consequences follow. First, the number of leaves is the product of the window
widths of the outer $N-1$ operators. The innermost operator only sets the window
bounds of each leaf and does not itself branch. Second, every leaf is resolved by
the horizon $H$, which is the sum of the upper bounds of all operators. In the running example, the outer $\square$ produces three
leaves for $i_1 \in \{2,3,4\}$. Each is a copy of $\diamondsuit_{[3,5]}$ shifted to
the absolute window $[i_1{+}3,\,i_1{+}5]$, and the horizon is $H = 4+5 = 9$.

\subsection{The \texttt{until\_leaf} primitive}

A leaf watches a postcondition \texttt{phi} over a fixed response window
$[\mathtt{lo},\mathtt{hi}]$, subject to a precondition \texttt{psi} that is
required to hold from the leaf's start tick $\mathtt{i\_start}$ onward. The
precondition is the ``alive'' signal that a leaf inherits when it sits under a
parent until; for a plain $\square$ or $\diamondsuit$ leaf there is no precondition
and \texttt{psi} is set to \texttt{true}. The leaf returns the three-valued verdict
as a pair \texttt{(rp, rn)} for \textsc{def\_true}/\textsc{def\_false}, generated as the
node in \Cref{lst:until-leaf}.

\begin{figure}[htpb]
\begin{lstlisting}[style=lustre,caption={The \texttt{until\_leaf} primitive (generated Lustre).},label={lst:until-leaf}]
node until_leaf(phi, psi: bool; lo_c, hi_c, i_start_c: int)
  returns (rp, rn: bool);
var n: int; in_win, in_active, alive, found: bool;
let
  n         = 0 -> pre(n) + 1;
  in_win    = (n >= lo_c) and (n <= hi_c);
  in_active = n >= i_start_c;
  alive     = true -> (pre(alive) and (not pre(in_active) or pre(psi)));
  found     = false -> (pre(found) or (in_win and phi and alive));
  rp        = found;
  rn        = not found and (not alive or (n >= hi_c));
tel
\end{lstlisting}
\end{figure}

The counter \texttt{n} records ticks since the leaf started, and \texttt{in\_win}
marks the response window. \texttt{alive} tracks whether \texttt{psi} has held on
every tick since the window became active (always true for $\square$/$\diamondsuit$
leaves, where \texttt{psi}$=$\texttt{true}), and \texttt{found} latches on the first
in-window tick where \texttt{phi} holds while alive. The outputs read off these:
\texttt{rp} (\textsc{def\_true}) is \texttt{found}, and \texttt{rn}
(\textsc{def\_false}) fires once the obligation cannot be met, either because the
precondition broke (\texttt{alive} false) or the window closed with no witness
($\mathtt{n}\ge\mathtt{hi\_c}$). These are the early-termination and window-expiry
cases of \Cref{def:explicit-semantics}; \Cref{lem:leaf-correct} formalises this

This one node covers all three operators. By \Cref{def:derived-ops},
$\diamondsuit_{[a,b]}\,\phi$ is an until with no precondition, so we use the leaf with \texttt{psi}$=$\texttt{true}, where a single in-window witness suffices, and $\square_{[a,b]}\,\phi$ is its dual, obtained through the swap given by
\Cref{def:three-valued} by feeding $\neg\phi$ to the leaf and exchanging \texttt{rp}
and \texttt{rn}.

\subsection{Composing the copies into an observer}

\noindent\begin{minipage}{\linewidth}
\begin{lstlisting}[style=lustre]
node Chain_G2_4_F3_5(phi: bool) returns (def_true,def_false: bool);
var rp_2,rn_2, rp_3,rn_3, rp_4,rn_4: bool; let
  (rp_2, rn_2) = until_leaf(phi, true, 5, 7, 2);
  (rp_3, rn_3) = until_leaf(phi, true, 6, 8, 3);
  (rp_4, rn_4) = until_leaf(phi, true, 7, 9, 4);
  def_true  = rp_2 and rp_3 and rp_4;
  def_false = rn_2 or rn_3 or rn_4;
tel
\end{lstlisting}
\end{minipage}

The observer above is the running example
$\square_{[2,4]}(\diamondsuit_{[3,5]}\,\phi)$: one leaf per start tick
$i\in\{2,3,4\}$, their verdicts combined with \texttt{and} on the positive side and
\texttt{or} on the negative side.

The general pattern follows the same shape. The observer is built bottom-up, one Lustre
node per sub-formula, so the node hierarchy mirrors the formula tree and every node
exposes the same $(\mathtt{def\_true},\mathtt{def\_false})$ pair. Each operator node
aggregates its children's verdicts by the Kleene rule for its connective ($\square$
conjoins positives and disjoins negatives, $\diamondsuit$ dual). Since $\square$ and
$\diamondsuit$ carry no precondition, every child contributes at every tick, so the
aggregation has no timing dependency.

An intermediate $\mathbf{U}$ is the one operator that cannot be directly realised by a simple composition of leaves. A branch $j$ is a valid witness only if the precondition $\psi$ held continuously from the branch's start up to tick $j-1$, and this is known only once tick $j$ is reached.
Selecting on the live $\psi$ stream before then gives the wrong verdict. We instead
track, for each branch, a running alive bit $\mathtt{av} = \texttt{true} \rightarrow
(\texttt{pre}(\mathtt{av}) \wedge (\neg\texttt{pre}(n\ge i_{\mathit{start}}) \vee
\texttt{pre}(\psi)))$, and freeze it at tick $j$ into a latch $\mathtt{latch}_j =
\texttt{true} \rightarrow (\textbf{if } n>j \textbf{ then } \texttt{pre}
(\mathtt{latch}_j) \textbf{ else } \mathtt{av})$. The aggregated verdicts are
$\mathtt{ap} = \bigvee_j (\mathtt{latch}_j \wedge \mathtt{inner\_pos}_j)$ and
$\mathtt{an} = \bigwedge_j (\neg\mathtt{latch}_j \vee \mathtt{inner\_neg}_j)$.

The whole construction is a single recursion, given compactly as \Cref{alg:synth}.

\begin{algorithm}[tbp]
\caption{Bounded-observer synthesis. \textsc{Synth} returns the
$(\mathit{pos},\mathit{neg})$ verdict wires of a formula, recursing on its structure.}
\label{alg:synth}
\begin{algorithmic}[1]
\Procedure{Synth}{formula $\psi$}
  \If{$\psi$ is a negation $\neg\chi$}
    \State \Return \Call{Synth}{$\chi$} with its two verdicts swapped
  \ElsIf{$\psi$ is an innermost $\square$, $\diamondsuit$ or $\mathbf{U}$ over a predicate}
    \State \Return one \textsc{until\_leaf} watching that predicate over the window
  \ElsIf{$\psi$ is $\square$ or $\diamondsuit$ over a sub-formula $\chi$}
    \For{each start tick $i$ in the window}
      \State $(p_i,n_i) \gets \Call{Synth}{\chi}$ shifted to start at $i$
    \EndFor
    \State \Return them combined: $\square$ needs \emph{all} $p_i$ (\emph{any} $n_i$);
            $\diamondsuit$ is the dual
  \ElsIf{$\psi$ is an until $\chi_1\,\mathbf{U}\,\chi_2$ over a sub-formula}
    \State as for $\diamondsuit$ on $\chi_2$, but keep a copy only while its precondition
           $\chi_1$ still holds
    \State \Return the combined verdict
  \ElsIf{$\psi$ is $\chi_1\wedge\chi_2$ or $\chi_1\vee\chi_2$}
    \State \Return \Call{Synth}{$\chi_1$} and \Call{Synth}{$\chi_2$} combined by the Kleene rule
  \EndIf
\EndProcedure
\end{algorithmic}
\end{algorithm}

\subsection{Correctness}
By \Cref{rem:obligations} the synthesised observer is correct once two obligations
are met: (i) each leaf emits the explicit verdicts of \Cref{def:explicit-semantics}
at every tick, and (ii) substituting a child's $(\mathit{pos},\mathit{neg})$ pair for
an atom is sound, so nesting composes. We discharge (i) as \Cref{lem:leaf-correct} and
(ii) as \Cref{lem:substitution}, using \Cref{lem:convergence} as support, and then
combine them by induction on nesting depth in \Cref{thm:soundness-completeness}. The
Kleene mutual-exclusion invariant of Property~\ref{prop:kleene-invariant} is immediate
and needs no separate argument: \Cref{lem:leaf-correct} gives it at each leaf, since
\texttt{rn} carries the conjunct $\neg\mathtt{rp}$, and every aggregation node is a
Kleene $\wedge$/$\vee$ of mutually exclusive child pairs, which is again mutually
exclusive. Full proofs are given in \ifextended\Cref{app:proofs}.\else
\suggest{the extended version.\footnote{\suggest{\url{TODO-EXTENDED-VERSION-URL}}}}\fi

\begin{lemma}[Leaf Correctness; obligation (i)]\label{lem:leaf-correct}
The \texttt{until\_leaf} node of \Cref{lst:until-leaf} realises the explicit
semantics of a single bounded operator. At every tick exactly one of
\textsc{def\_true} ($\mathtt{rp}$), \textsc{def\_false} ($\mathtt{rn}$), or
\textsc{unknown} ($\neg\mathtt{rp}\wedge\neg\mathtt{rn}$) holds, and $\mathtt{rp}$
(resp.\ $\mathtt{rn}$) is true exactly when $\llbracket\cdot\rrbracket^{+}$ (resp.\
$\llbracket\cdot\rrbracket^{-}$) of \Cref{def:explicit-semantics} holds for that
operator over its window.
\end{lemma}

\begin{lemma}[Convergence and Causality]\label{lem:convergence}\label{lem:causality}
Every verdict is monotone, changing at most once from \textsc{unknown} to a definitive value and constant thereafter, and is fixed from the horizon tick $H=\sum_k b_k$ onward where $b_k$ is the upper interval bound of the $k$-th operator $\mathit{op}_k^{[a_k,b_k]}$ in the chain and the sum runs over all $N$ nested operators. Consequently the composition is causal: a parent aggregates only child verdicts that are already definitive when its own verdict becomes definitive.
\end{lemma}

\begin{lemma}[Compositional Substitution; obligation (ii)]\label{lem:substitution}
Let $\varphi$ have an immediate child $\chi$ occupying the position of an atom. If the
child observer emits the explicit verdicts of $\chi$, that is, its pair
$(\mathit{pos}_\chi,\mathit{neg}_\chi)$ equals
$(\llbracket\chi\rrbracket^{+},\llbracket\chi\rrbracket^{-})$ at every tick, then
substituting that pair for the atom's $(p_t,\neg p_t)$ in the enumerations of
\Cref{def:explicit-semantics} yields exactly
$(\llbracket\varphi\rrbracket^{+},\llbracket\varphi\rrbracket^{-})$.
\end{lemma}

\begin{theorem}[Soundness and Completeness]\label{thm:soundness-completeness}
For every bounded SSTL formula $\varphi$ and every tick, the synthesised observer's
verdict pair $(\mathit{pos},\mathit{neg})$ equals the explicit three-valued
semantics of \Cref{def:explicit-semantics}. That is, $\mathit{pos}$ is true if and
only if $\llbracket\varphi\rrbracket^{+}$, and $\mathit{neg}$ is true if and only if
$\llbracket\varphi\rrbracket^{-}$.
\end{theorem}

\paragraph{Boolean connectives between sub-formulas.}

Many useful properties combine several predicates, for example
$\square_{[2,4]}(\diamondsuit_{[3,5]}\,\mathit{rc}\wedge\diamondsuit_{[2,6]}\,
\mathit{weps})$. The Boolean connectives extend the synthesis at no real cost: each is one more node
aggregating its children's $(\mathit{pos},\mathit{neg})$ pairs by the Kleene rules
($A\wedge B$ sets $\mathit{pos}=\mathit{pos}_A\wedge\mathit{pos}_B$,
$\mathit{neg}=\mathit{neg}_A\vee\mathit{neg}_B$; $A\vee B$ dual), leaving
\texttt{until\_leaf} unchanged. Unlike a temporal operator, which multiplies its
child's leaf count by $(b-a+1)$, a connective is additive, with $L(A)+L(B)$ leaves, so
a compound observer is no larger than its component chains wired together.

\section{Unbounded Outer $\square$ via Shift Register}
\label{sec:shift}

The synthesis so far assumes every operator is bounded, monitoring only a fixed horizon
from the start of the run, which is unsuitable for free-running systems. The unbounded
globally operator $\square$ is a common and intuitive specification pattern for such
systems.

A genuinely unbounded $\square(\mathit{inner})$ asserts that the bounded inner formula
holds when started at every tick, indefinitely; conceptually an infinite conjunction
of inner-observer copies started at ticks $0, 1, 2, \dots$. We cannot instantiate
infinitely many observers, but we need not: by \Cref{lem:convergence} each copy resolves
within the inner horizon $H$ (the sum of the inner operators' upper bounds) and is
constant afterwards. So at any tick only the $H{+}1$ copies started within the last $H$
ticks are still accumulating; older copies have issued their verdict, and newer ones do
not yet exist. This bounded working set is what a shift register captures.

Instead of allocating a new observer each tick, we keep a shift register of $H{+}1$
slots that recycles a fixed amount of state across all start ticks. Slot $k$ holds the
running state of the inner copy that started $k$ ticks ago. Each tick every copy ages
by one, so the slots shift forward: slot $k$ inherits slot $k{-}1$'s state from the
previous tick. Slot $0$ is the constant \texttt{false} and serves as the initial state
of each newly started copy, and a copy that reaches slot $H$ has lived its full horizon
and yields its final verdict.

The register is maintained per leaf of the inner observer. Each inner leaf contributes
an accumulator $\mathtt{f\_k}$ at slot $k$ that carries slot $k{-}1$'s value forward and
adds the leaf's $\phi$-evidence whenever the copy is at step $k$ of that leaf's window,
\[
  \mathtt{f\_0} = \mathit{false},\qquad
  \mathtt{f\_k} = \mathit{false} \to
    \mathtt{pre}(\mathtt{f\_{k-1}}) \;\mathbf{or}\;
    (lo \le k \le hi \;\mathbf{and}\; \phi_{\mathit{leaf}}).
\]
At slot $H$ the leaf accumulators are combined by the same tree aggregation as the
bounded inner observer, giving the verdict of the copy that started $H$ ticks ago.

The unbounded $\square$ reuses the bounded recursion of \Cref{alg:synth} with two
changes, summarised in \Cref{alg:shift}: each leaf becomes an $(H{+}1)$-slot recycling
register, and the outer $\square$ becomes a single latch over the inner negative
verdict.

\begin{algorithm}[tbp]
\caption{Unbounded $\square$ by recycling the bounded observer of \Cref{alg:synth}.}
\label{alg:shift}
\begin{algorithmic}[1]
  \State Keep a shift register of $H{+}1$ slots; slot $k$ holds the copy started $k$
         ticks ago.
  \For{each tick}
    \State start a fresh copy in slot $0$ and shift every copy one slot forward, recycling the slot that falls off the end
    \State the copy now in slot $H$ has run its full horizon: combine its leaves as in \Cref{alg:synth} for its verdict
    \State if that verdict is false, $\square$ is false forever (latched); otherwise it stays unknown
  \EndFor
\end{algorithmic}
\end{algorithm}

On an infinite run the only definitive verdict is a latching FALSE: once any copy violates, the property is false forever, and it can never be confirmed TRUE on a finite prefix. The per-window verdicts are thus often more informative.

\section{Visualising Nested Observers}
\label{sec:visualiser}

A nested formula such as $\square_{[2,4]}(\diamondsuit_{[3,5]}(\square_{[1,3]}\,p))$
has no single horizon a designer can reason about by inspection: each outer offset
spawns its own inner window, the windows overlap, and the streaming verdict can
close at a tick that bears no obvious relation to the syntactic bounds (the
early-termination and window-expiry modes of \Cref{def:explicit-semantics}). A
pass/fail value per tick is not enough; the designer needs to see which copy of the
outer operator carried the verdict and why it closed when it did.

We provide a standalone, browser-based visualiser for the synthesised observers.
The designer enters a formula and an input trace; the tool synthesises the observer
in-browser using the same engine as above, runs it over the trace, and renders a
per-copy Gantt chart in which each row is one copy of the outer operator and each
cell carries that copy's three-valued verdict at that tick.

Formula entry accepts the full grammar (bounded $\square/\diamondsuit/\mathbf{U}$,
the Boolean connectives, and $\neg$), with named presets for the case-study properties,
and Boolean input streams are toggled per tick directly on the chart. Reading down a
column gives the aggregate verdict, and reading across a row shows one copy's window
opening, waiting, and closing (\Cref{fig:visualiser}). A \texttt{Compare lv6} action
runs the synthesised Lustre through the actual \texttt{lv6} compiler and overlays its
verdict for cross-checking. The tool is for design-time intuition and debugging, not
verification, and is available as part of this paper's artefact.\footnote{\artefacturl}

\begin{figure}[htpb]
  \centering
  \includegraphics[width=0.5\linewidth]{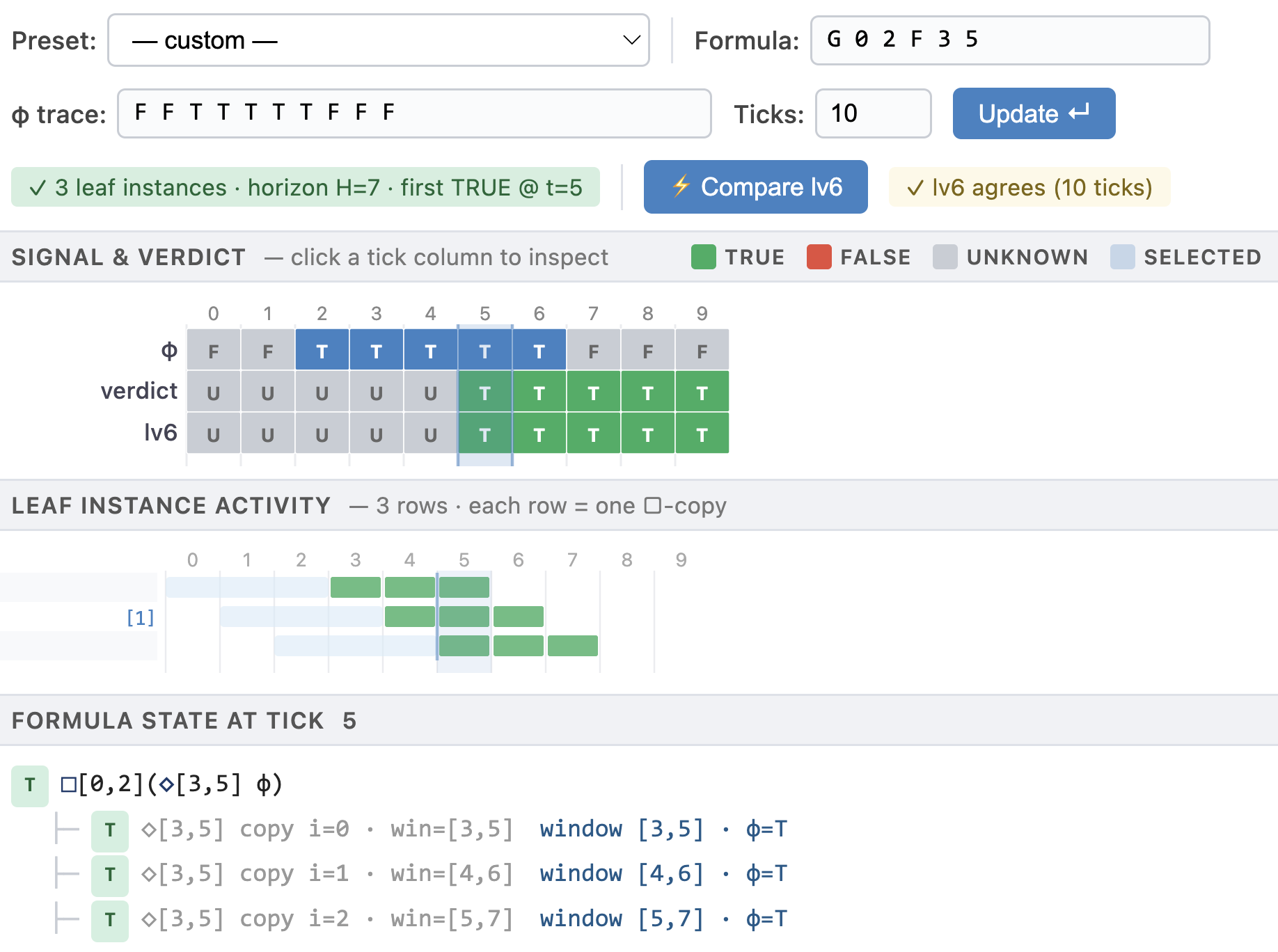}
  \caption{The browser-based visualiser rendering the per-copy three-valued verdict
    of a nested formula over an editable input trace.}
  \label{fig:visualiser}
\end{figure}

\section{Evaluation}
\label{sec:evaluation}
We synthesised all eight observers with \texttt{synthesize.py} and ran them online
over the two case studies with lv6, the Lustre interpreter. Because the
observers are ordinary Lustre, they are also amenable to static verification: as a
cross-check we discharged each property with the model checker Kind2, and the runtime
and static verdicts agree on every property, including the mutual-exclusion invariant
$\mathit{pos}\wedge\mathit{neg}=\mathit{false}$ holding at every tick for every
observer.
\Cref{tab:observers} collects each observer's size and the tick at which it issues
a definitive verdict, and \Cref{tab:trace} gives the spring-mass simulation trace.

\begin{table}[htbp]
\centering
\caption{Observer size, verdict timing, and per-tick execution cost. Leaves $=$ number
  of \texttt{until\_leaf} instances; $H$ $=$ convergence horizon (ticks); state $=$
  Boolean registers}
\label{tab:observers}
\smallskip\small
\begin{tabular}{|l|l|r|r|r|l|r|}
\hline
Prop. & SSTL formula & Leaves & $H$ & State & Verdict & ns/tick \\
\hline
\multicolumn{7}{|l|}{\textit{Spring-mass system}} \\
P1    & $\square_{[2,4]}(\diamondsuit_{[3,5]}\,\mathit{reach})$
      & 3  & 9  & $\approx 10$  & \textbf{T} @7 & 8.2 \\
P2    & $\square(\diamondsuit_{[5,20]}(\square_{[1,5]}\,\mathit{weps}))$
      & 16 & 25 & 416 & none$^\dagger$ & 95.5 \\
P2$'$ & $\square(\diamondsuit_{[3,5]}(\square_{[1,2]}\,\mathit{wstrict}))$
      & 3  & 7  & 24  & \textbf{F} @7$^\dagger$ & 6.7 \\
P3    & $\mathit{mono}\,\mathbf{U}_{[3,15]}(\diamondsuit_{[1,5]}(\square_{[1,3]}\,\mathit{weps}))$
      & 65 & 23 & $\approx 143$ & \textbf{T} @8 & 206.5 \\
\hline
\multicolumn{7}{|l|}{\textit{Car-following (ACC)}} \\
Q1    & $\square_{[2,6]}(\diamondsuit_{[1,4]}\,\mathit{safe\_gap})$
      & 5  & 10 & $\approx 10$  & \textbf{T} @7 & 13.1 \\
Q2    & $\square(\diamondsuit_{[3,10]}(\square_{[1,4]}\,\mathit{safe\_gap}))$
      & 8  & 14 & 120 & none$^\dagger$ & 25.8 \\
Q2$'$ & $\square(\diamondsuit_{[1,4]}(\square_{[1,2]}\,\mathit{adeq\_gap}))$
      & 4  & 6  & 28  & \textbf{F} @14$^\dagger$ & 7.4 \\
Q3    & $\mathit{vel\_moving}\,\mathbf{U}_{[2,10]}(\diamondsuit_{[1,3]}(\square_{[1,2]}\,\mathit{gap\_settled}))$
      & 27 & 15 & $\approx 63$ & \textbf{T} @12 & 85.4 \\
\hline
\multicolumn{7}{|l|}{\footnotesize $^\dagger$ Unbounded $\square$ has no finite formula
  horizon}
\end{tabular}
\end{table}

\begin{table}[htbp]
\centering
\caption{Spring-mass trace. Predicates: reach ($rc$), within\_eps
  ($we$), within\_strict ($ws$), monotone ($mn$). Bold marks
  first transitions.}
\label{tab:trace}
\smallskip\small
\begin{tabular}{r@{\hspace{4pt}}r@{\hspace{6pt}}cccc@{\hspace{6pt}}cccc}
\hline
$t$ & $\mathit{pos}$ & $rc$ & $we$ & $ws$ & $mn$ & P1 & P2 & P2$'$ & P3 \\
\hline
 0 & 0.000 & F & F & F & T & U & U & U & U \\
 5 & 0.190 & F & F & F & T & U & U & U & U \\
 6 & 0.247 & F & T & F & T & U & U & U & U \\
 7 & 0.297 & T & T & T & T & \textbf{T} & U & \textbf{F} & U \\
 8 & 0.337 & T & T & F & T & T & U & F & \textbf{T} \\
11 & 0.387 & F & T & F & T & T & U & F & T \\
12 & 0.384 & F & T & F & F & T & U & F & T \\
\hline
\end{tabular}
\end{table}

The P2 and P2$'$ columns report the per-window diagnostic of the
unbounded $\square$: \textbf{F} at the tick a window closes violated, and U while no
window has yet violated. An unbounded $\square$ can only ever be falsified, never
confirmed, on a finite prefix, so these columns never reach \textbf{T}.

To assess runtime feasibility, we measured the per-tick cost of each observer compiled
to C through the \texttt{lv6} back end. Since the observer is a fixed dataflow graph
evaluated in full each tick, the per-tick cost is constant; we report the minimum over
five $3\times10^{7}$-tick runs, as pure observer compute (no I/O or interpreter
overhead) on a 2020 MacBook Pro, in the final (\textsc{ns/tick}) column of
\Cref{tab:observers}. Every observer runs in well under a microsecond per tick
(spring-mass averaging $79.2$\,ns, car-following $32.9$\,ns). The cost tracks leaf count
rather than nesting depth, so the deeply unrolled $\mathbf{U}$-rooted chains (P3, Q3)
are the most expensive while tight bounded properties cost only a few nanoseconds. Even the slowest
observer leaves a wide margin below a typical millisecond-scale control period,
confirming that the synthesised monitors are cheap enough to run online.

The central claim, that the three-valued semantics delivers a definitive verdict
before the formula horizon, holds across both case studies: every bounded property
settles strictly before its horizon $H$ (\Cref{tab:observers}). The margin is largest
for the most deeply nested property, P3, which settles \textbf{T} at tick~$8$ against a
horizon of $H=23$, fifteen ticks early.

The unbounded-$\square$ properties behave as expected for a falsification-only
operator. The genuinely safe properties P2 and Q2 record no window violation over the
entire run, while the deliberately tight variants are caught: P2$'$ latches false at
tick~$7$ when the strict settling window first closes unsatisfied, and Q2$'$ latches
false at tick~$14$, flagging the gap dropping below the adequate threshold during the
braking transient. In each case the monitor reports the violation at the tick the
offending window closes, never later. All verdicts shown were produced by running the
synthesised observers under the \texttt{lv6} Lustre interpreter.

\section{Related Work}
\label{sec:related}
The closest precursor is Bellanger et al.~\cite{bellanger2025formally}, who synthesise Lustre nodes for non-nested STL operators using Kleene three-valued logic and verify them with Kind2. They share our synchronous co-compilation setting and
verdict structure, but explicitly exclude nested operators such as
$\square_{[a,b]}(\diamondsuit_{[c,d]}\,\phi)$. Their observers are moreover
oriented toward static verification as Kind2 contracts, whereas we treat the observer
primarily as an executable specification monitored online, with static checking a
by-product. Our
\texttt{until\_leaf} primitive supports arbitrary $N$-deep nesting of $\square$,
$\diamondsuit$, and $\mathbf{U}$, and the shift-register construction extends the outermost operator to a truly unbounded $\square$.

Roehm et al.~\cite{roehm2016stl} perform STL model checking on hybrid automata, which is sound but incomplete.
Bae and Lee~\cite{bae2019bounded} give refutationally complete bounded STL model checking by syntactic separation, supporting arbitrary nesting, with an SMT encoding size that is exponential in the until-nesting depth $h$ and the signal variability bound $k$. STLmc~\cite{yu2022stlmc} builds on this engine, reducing robust STL model checking to the Boolean problem of~\cite{bae2019bounded}. These SMT tools verify a separately supplied model, while our observers are co-compiled with the Lustre source and produce tick-by-tick verdicts at runtime with no modelling gap.
Donz\'e et al.'s efficient robust monitoring~\cite{donze2013efficient} and
S-TaLiRo~\cite{annpureddy2011s} analyse temporal properties over time-series data, the
former computing robustness streams and the latter performing stochastic falsification. Deshmukh et
al.~\cite{deshmukh2017robust} give three-valued online STL monitoring via robustness degree, where the third value arises from a quantitative distance to violation; this
is orthogonal to our causal three-valued semantics, where Unknown means the formula
horizon has not yet elapsed. Both tools operate on pre-recorded traces, while our observers run live.

Reinbacher et
al.~\cite{reinbacher2012onchip,reinbacher2014runtime} synthesise past-time MTL
observers as reconfigurable FPGA blocks. Rather than unrolling a window of
width~$b$ in the fashion that we do, they maintain a garbage-collected
list of timestamp pairs recording the intervals over which a subformula
held, giving a \textit{Since} observer doubly logarithmic in time. Their logic is strictly
past-time, so our bounded-\emph{future} fragment lies outside it. However, their encoding is more memory efficient.
Jak\v{s}i\'c et al.~\cite{jaksic2015fpga} do handle bounded-future STL over
discrete time, compositionally and with nesting, as monitors on an FPGA that serve both monitoring and verification. However, they monitor by rewriting~$\varphi$ into an equisatisfiable past formula whose verdict for tick~$t$ is emitted at $t+H(\varphi)$, for a horizon~$H$. Their verdicts are two-valued, with no early termination.
Ulus~\cite{ulus2026sequential} likewise restricts to the past fragment, and
deliberately: a bounded-future verdict can only be produced by delaying the output (as we do, but only for as long as
  the trace leaves the verdict open), and
one arriving late ``may be too late'' when it feeds a reactive controller.
Although we share the bounded fragment with these works, none of them produces a monitor in the \emph{system's own language} that is simultaneously an executable observer and a model-checking obligation.

Closest in spirit to our online setting is RTAMT~\cite{yamaguchi2024rtamt}, a
library for online monitoring of STL that does run inside live
systems: it exposes a Python API with bindings for ROS and MATLAB/Simulink, and handles
bounded-future operators. Moreover, it supports robustness measurements for handling discretisation of continuous signals.
In contrast, our observers are co-compiled with the system in the same synchronous language, which also makes them amenable to static verification.

\section{Conclusion}
\label{sec:conclusion}

We presented a systematic synthesis of synchronous Lustre observers for arbitrarily
deep nested SSTL formulas. The construction rests on two main ideas: a
single \texttt{until\_leaf} primitive of which $\square$ and $\diamondsuit$ are
degenerate cases, and a shift-register construction that monitors a truly
unbounded outer $\square$. We proved the observers sound and
complete against the explicit three-valued semantics, and validated them on two case
studies with lv6 and Kind2. Additionally, we provide a web-based GUI for visualising nested SSTL formulae.

The approach has three main limitations, each pointing to future work. First, the leaf
count is the product of the window widths, so large windows (for example
$\diamondsuit_{[0,100]}$) cause the observer to blow up; sharing the overlapping leaves
or a symbolic enumeration would reduce this. Sibling leaves of one operator differ only in their window offset and read
the same child stream, so a leaf whose window is contained in another's is redundant, and the two could share state. Second, both case studies model the plant inside Lustre as a closed system, so
the observer reads internally generated signals rather than sampling an external
environment. Monitoring a genuine open system, one whose inputs arrive from sensors or
a co-simulated plant at runtime, would make the observers deployable as true online
monitors, and is a natural next step in the Lustre setting. Third, we are currently operating only within the sampled discretised domain: supporting a continuous STL robustness semantics requires the robustness degrees explored in prior work~\cite{fainekos2009robustness}.

\ifextended
\appendix
\section{Proofs}
\label{app:proofs}
\begin{proof}[of \Cref{lem:leaf-correct}]
Recall \Cref{lst:until-leaf}. The counter \texttt{n} is the number of
ticks since the leaf started. The window is $[\mathtt{lo\_c},\mathtt{hi\_c}]$, and \texttt{i\_start\_c} is the tick from which the precondition \texttt{psi} must hold.
We read each equation as a statement about the current tick \texttt{n}.
First, \texttt{alive} records whether the precondition has held since activation: it
is true at tick \texttt{n} exactly when \texttt{psi} was true at every tick from
\texttt{i\_start\_c} up to $\mathtt{n}{-}1$. This follows by induction on \texttt{n},
since \texttt{alive} starts true and its recurrence can only turn it false once the
leaf is active and \texttt{psi} failed on the previous tick. Second, \texttt{found}
is a latch: it becomes true at the first in-window tick where \texttt{phi} holds
while \texttt{alive} is true, and stays true thereafter. So \texttt{found} is true
exactly when a witness for \texttt{phi} has appeared inside the window during a
stretch where the precondition held throughout.

The verdict then splits into three cases at tick \texttt{n}. First, if a witness has
appeared, then \texttt{found}, and hence \texttt{rp}, is true. This is exactly
$\llbracket\cdot\rrbracket^{+}$, the disjunction over window positions of a witnessed
\texttt{phi} with \texttt{psi} maintained. Second, if no witness has appeared and
either the precondition has broken (\texttt{not alive}) or the window has closed
($\mathtt{n}\ge\mathtt{hi\_c}$), then \texttt{rn} is true. These are the
early-termination and window-expiry disjuncts of $\llbracket\cdot\rrbracket^{-}$.
Third, if no witness has appeared but the window is still open and the precondition
intact, both outputs are false and the verdict is \textsc{unknown}. The cases are
exhaustive. They are disjoint because \texttt{rn} carries the condition
\texttt{not found}, which is $\neg\mathtt{rp}$, so \texttt{rp} and \texttt{rn} are never true together. This gives agreement
with \Cref{def:explicit-semantics}.
\end{proof}

\begin{proof}[of \Cref{lem:convergence}]
A single leaf resolves by the close of its window. Its latch \texttt{found} stays
true once true, so \texttt{rp} is monotone, and if \texttt{found} is still false when
the window closes ($\mathtt{n}\ge\mathtt{hi\_c}$), then \texttt{rn} becomes true and
stays true, since \texttt{phi} can no longer be witnessed in a closed window. Each
leaf therefore changes at most once and is constant once its window has closed.
Aggregating monotone, eventually-constant inputs by \texttt{and} and \texttt{or} is
again monotone and eventually constant, and the last leaf to close does so by the
horizon $H=\sum_k b_k$, so the whole observer is fixed by $H$.
Causality follows. A child instance also closes its window by $H$, since its window bounds are sums of the enclosing operators' upper bounds, so the child verdict is already definitive when the parent reads it.
\end{proof}

\begin{proof}[of \Cref{lem:substitution}]
The enumerations of \Cref{def:explicit-semantics} reference the child only through its
verdict pair at the sampled offsets $t{+}i$. By \Cref{lem:causality} each such child
instance has closed its window and reached a definitive verdict by the tick the parent
samples it, so $\mathit{pos}_\chi$ and $\mathit{neg}_\chi$ at those offsets coincide
with the Boolean value of $\chi$ on every continuation, which is the property the
enumeration requires of its argument. Replacing $(p_t,\neg p_t)$ by
$(\mathit{pos}_\chi,\mathit{neg}_\chi)$ term-for-term therefore reproduces
$\llbracket\varphi\rrbracket^{\pm}$ exactly.
\end{proof}

\begin{proof}[of \Cref{thm:soundness-completeness}]
By induction on nesting depth, discharging the two obligations of
\Cref{rem:obligations}. At depth $1$ the formula is a single bounded operator over
atoms, and \Cref{lem:leaf-correct} (obligation (i)) gives observer equals explicit
semantics. Assume the claim for every observer of depth less than $d$, and let
$\varphi$ have depth $d$ with immediate children of smaller depth. By the hypothesis
each child observer equals $\llbracket\cdot\rrbracket^{\pm}$ of its subformula, so
\Cref{lem:substitution} (obligation (ii)) applies and the parent verdict equals
$\llbracket\varphi\rrbracket^{\pm}$. The two sides stay mutually exclusive throughout,
and by \Cref{lem:convergence} the verdict is
definitive by at latest tick $H$.
\end{proof}

\fi

\bibliographystyle{splncs04}
\bibliography{refs}

\end{document}